\documentclass[11pt,a4paper]{amsart}
\usepackage{a4wide}
\usepackage{filecontents} 
 
\usepackage{amsthm}
\usepackage{enumerate}
\usepackage{amssymb}
\usepackage{url,doi}
\usepackage{tikz}

\def\AntiFerroIsing#1{\textsc{AntiFerroIsing$_#1$}}
\def\AntiFerroIsingd#1#2{\textsc{AntiFerroIsing$_#1$}($\ensuremath{#2}$)}
\def\SAT{\mathrm{SAT}}
\def\nSAT{\mathrm{\#SAT}}
\def\nBIS{\mathrm{\#BIS}}
\def\RHPi{\mathrm{\#RH}\Pi_1}
\def\bfS{\mathbf{S}}
\def\listHcol{\listXcol{H}}
\def\listHcold{\listXcold{H}{\ensuremath{\Delta}}}
\def\IBPlistHcol{\listHcol}
\def\RPIlistHcol{\listHcol}
\def\listXcol#1{\textsc{\#List-$#1$-Col}}
\def\listXcold#1#2{\textsc{\#List-$#1$-Col}($\ensuremath{#2}$)}
\def\Xcol#1{\textsc{\#$#1$-Col}}
\def\Xcold#1#2{\textsc{\#$#1$-Col}($\ensuremath{#2}$)}
\def\Hcol{\Xcol{H}}
\def\Hcold{\Xcold{H}{\Delta}}
\def\pnsat{\textsc{\#1p1nSat}}

\def\bfx{\mathbf{x}}

\def\submat#1#2#3#4#5{#1_{(#2,#3),(#4,#5)}}
\def\IMP{\mathrm{IMP}}
\def\Ptime{\mathrm{FP}}    
\def\numP{\#\mathrm{P}}
\def\NP{\mathrm{NP}}
\def\RP{\mathrm{RP}}

\def\firstLetH{Suppose that $H$ is a connected undirected graph (possibly with loops). }
\def\LetH{Suppose that $H$ is a connected undirected graph. }

\tikzset{lab/.style={circle,draw,inner sep=0pt,fill=none,minimum size=5mm}}
\tikzset{every picture/.style={line width=0.8pt}}
\tikzset{empty/.style={rectangle,draw=none,fill=none}}

\newtheorem{theorem}{Theorem}
\newtheorem{lemma}[theorem]{Lemma}
\newtheorem{corollary}[theorem]{Corollary}
\newtheorem{definition}[theorem]{Definition}
\theoremstyle{remark}
\newtheorem*{remarks}{Remarks}
\newtheorem*{remark}{Remark}

\begin{document}

\title{A complexity trichotomy for approximately counting list $H$-colourings$^\ddag$$^\S$}

\thanks{$\ddag$ These results were presented in preliminary form in the proceedings of \textit{ICALP 2016} (Track A)}
\author{Andreas Galanis$^{*}$}\thanks{$*$ The research leading to these results 
has received funding from the European Research Council under 
the European Union's Seventh Framework Programme (FP7/2007-2013) 
ERC grant agreement no.\ 334828. The paper reflects only the authors' views 
and not the views of the ERC or the European Commission. 
The European Union is not liable for any use that may be made of the information contained therein.}
\address{Andreas Galanis, Department of Computer Science, University of Oxford, 
Wolfson Building, Parks Road, Oxford, OX1~3QD, UK.}
\author{Leslie Ann Goldberg$^{*}$}
\address{Leslie Ann Goldberg, Department of Computer Science, University of Oxford, 
Wolfson Building, Parks Road, Oxford, OX1~3QD, UK.}
\author{Mark Jerrum$^{\dag}$}\thanks{$\dag$ This work was partially supported by the EPSRC grant EP/N004221/1.}
\address{Mark Jerrum, School of Mathematical Sciences\\
Queen Mary, University of London, Mile End Road, London E1 4NS, United Kingdom.}
\thanks{$\S$ This work was done in part while the authors were visiting the Simons Institute for the Theory of Computing.}
\date{January 5, 2017}
            
\begin{abstract}
We examine the computational complexity of 
approximately counting the list $H$-colourings of a graph. 
We
discover a natural graph-theoretic trichotomy based on the structure of the
graph~$H$. 
If $H$ is an irreflexive bipartite graph or a reflexive complete graph then  
counting list $H$-colourings is trivially in polynomial time.
Otherwise, if $H$ is an irreflexive bipartite permutation graph or a reflexive proper interval graph then 
approximately counting list $H$-colourings is equivalent to $\nBIS$,  the  
problem of approximately counting independent sets in a bipartite graph. This is a well-studied problem which is believed to be of intermediate complexity -- it is believed that it does not have an FPRAS, but
that it is not as difficult as approximating the most difficult counting problems in $\numP$.
For every other graph~$H$, approximately counting list $H$-colourings is complete for $\numP$ with respect to
approximation-preserving reductions (so there is no FPRAS unless $\NP=\RP$). 
Two pleasing features of
the trichotomy are 
(i)  it has a natural formulation in terms of hereditary graph classes, and (ii) the proof is largely self-contained and does not require any universal algebra (unlike similar dichotomies in the weighted case).
We are able to extend the hardness results to the bounded-degree setting,
showing that all hardness results apply to input graphs with 
maximum degree at most~$6$.
 \end{abstract}

\maketitle

\section{Overview}\label{sec:intro}

In this paper we study the complexity of approximately counting 
the list $H$-colourings of a graph.
List $H$-colourings generalise $H$-colourings in the same way that 
list colourings generalise  proper vertex colourings.  
Fix an undirected graph $H$, which may have loops but not parallel edges.
Given a graph $G$, an {\it $H$-colouring of $G$} is a homomorphism from $G$ to~$H$ ---
that is, a mapping $\sigma:V(G)\to V(H)$
such that,
for all $u,v\in V(G)$, $\{u,v\}\in E(G)$ implies $\{\sigma(u),\sigma(v)\}\in E(H)$.  
If we identify the vertex set $V(H)$ with a set 
$Q=\{1,2,\ldots,q\}$ of ``colours'', then we can think of the mapping $\sigma$ 
as specifying a colouring of the vertices~$G$, and we can interpret the 
graph~$H$ as specifying the allowed colour adjacencies:  adjacent 
vertices in~$G$ can be assigned colours $i$ and $j$, if and only if 
vertices $i$ and~$j$ are adjacent in~$H$.  

 Now consider the graph $G$ together with a collection 
of sets $\bfS=\{S_v\subseteq Q:v\in V(G)\}$ specifying allowed colours 
at each of the vertices.  A {\it list $H$-colouring of $(G,\bfS)$} is 
an $H$-colouring $\sigma$ of~$G$ satisfying $\sigma(v)\in S_v$, for all $v\in V$. 
In the literature, the set~$S_v$ is referred to as the ``list'' of allowed colours at vertex~$v$, but
there is no implied ordering on the elements of~$S_v$ --- $S_v$ is just a set of allowed colours.

Suppose that $H$ is a \emph{reflexive} graph (i.e., a graph in which each vertex has a loop).
Feder and Hell~\cite{FH} considered the complexity of determining whether
a list $H$-colouring exists, given an input $(G,\bfS)$.
They showed that the problem is in~$\Ptime$ if $H$ is an interval graph, and that it is $\NP$-complete, otherwise.
Feder, Hell and Huang~\cite{FHH} studied the same problem in the case where $H$ is
\emph{irreflexive} (i.e., $H$ has no loops).
They showed that the problem is in~$\Ptime$ if $H$ is a circular arc graph of clique covering number two
(which is the same as being the complement of an interval bigraph~\cite{HellHuang}),
and that it is $\NP$-hard, otherwise. Finally, Feder, Hell and Huang~\cite{FHH03} 
generalised this result to obtain a dichotomy for all~$H$.
They introduced a new class of graphs, called bi-arc graphs, and showed that the problem is in~$\Ptime$ if $H$
is a bi-arc graph, and $\NP$-complete, otherwise.

We are concerned with the computational complexity of counting list $H$-colourings. 
Specifically  we are interested in how the complexity of the following 
computational problem depends on~$H$.   
\begin{description}
\item[Name] $\listHcol$.
\item[Instance]  A graph $G$   and a collection of colour sets 
$\bfS=\{S_v\subseteq Q:v\in V(G)\}$, where $Q=V(H)$.
\item[Output] The number of list $H$-colourings of $(G,\bfS)$.
\end{description}
Note that it is of no importance whether we allow or disallow loops in~$G$
--- a loop at vertex $v\in V(G)$ can be encoded within the set $S_v$ --- 
so we adopt the convention that $G$ is loop-free.  
As in the case of the decision problem,
$H$ is a parameter of the problem --- it does not form part of the problem instance.
Sometimes we obtain sharper results by introducing an additional parameter~$\Delta$,
which is an upper bound on the degrees of the vertices of~$G$.
\begin{description}
\item[Name] $\listHcold$.
\item[Instance]  A graph $G$  with maximum degree at most~$\Delta$ and a collection of colour sets 
$\bfS=\{S_v\subseteq Q:v\in V(G)\}$, where $Q=V(H)$.
\item[Output] The number of list $H$-colourings of $(G,\bfS)$.
\end{description}

Although $\listHcol$ and $\listHcold$ are the main objects of study in this paper, we occasionally need to 
discuss the more basic versions of these problems without lists.
\begin{description}
\item[Name] $\Hcol$.
\item[Instance]  A graph $G$.
\item[Output] The number of $H$-colourings of $G$.
\end{description}

\begin{description}
\item[Name] $\Hcold$.
\item[Instance]  A graph $G$  with maximum degree at most~$\Delta$.
\item[Output] The number of $H$-colourings of $G$.
\end{description}  

To illustrate the definitions, 
let $K'_2$ be the first graph illustrated in Figure~\ref{fig:2-wrench},
consisting of two connected vertices with  a loop on vertex~$2$.
$\Xcol{K_2'}$ 
is the problem of counting independent sets
in a graph since the vertices mapped to  colour~$1$ by any homomorphism form an independent set.
Let $K_3$ be the complete irreflexive graph on three vertices. Then $\Xcol{K_3}$  is the problem of counting  the proper 3-colourings
of a graph. $\Xcold{K_2'}{\Delta}$ and $\Xcold{K_3}{\Delta}$ are the corresponding problems where the instance is restricted to have maximum degree at most~$\Delta$.

The computational complexity of computing exact solutions to 
$\Hcol$ and
$\Hcold$
was  determined by Dyer and Greenhill~\cite{DG00}.
Dyer and Greenhill showed that 
$\Hcol$ is in $\Ptime$ if $H$ is a complete reflexive graph  
or a complete bipartite irreflexive graph, and 
$\Hcol$ is $\numP$-complete otherwise.
Their dichotomy also extends to the bounded-degree setting.
In particular, they showed that if $H$ is not a complete reflexive graph or 
a complete bipartite irreflexive graph then
there is an integer~$\Delta_H$ such that, for all $\Delta \geq \Delta_H$,
$\Hcold$ is $\numP$-complete.

Since the polynomial-time cases in Dyer and Greenhill's dichotomy clearly remain 
solvable in polynomial-time in the presence
of lists, their dichotomy for $\Hcol$ carries over to $\listHcol$  without change.  
In other words, there is no difference between the complexity of 
$\Hcol$ and $\listHcol$ as far as exact computation is concerned.
However, this situation changes    
if we consider approximate 
counting,
and this is the phenomenon that we explore in this paper.

With a view to reaching the statement of the main results as quickly as possible,
we defer precise definitions of the relevant concepts to Section~\ref{sec:prelim},
and provide only indications here.  From graph theory we import a couple of
well studied hereditary graph classes, namely
bipartite permutation graphs and proper interval graphs.  These classes each have 
several equivalent characterisations, and we give two of these, namely, excluded 
subgraph and matrix characterisations, in Section~\ref{sec:prelim}.  
It is sometimes useful
to restrict the definition of proper interval graphs to simple graphs.
However, in this paper, 
as in \cite{FH},
we consider reflexive proper interval graphs. 

From complexity theory we need the definitions of a Fully Polynomial Randomised 
Approximation Scheme (FPRAS), of approximation-preserving (AP-) reducibility,
and of the counting problems $\nSAT$ and $\nBIS$.  An FPRAS is a randomised algorithm that 
produces approximate solutions within specified relative error with high probability in 
polynomial time.  An AP-reduction from problem $\Pi$ to problem $\Pi'$ is a 
randomised Turing
reduction that yields close approximations to~$\Pi$ when provided with close approximations 
to~$\Pi'$.  It meshes with the definition of an FPRAS in the sense that the existence of an FPRAS for~$\Pi'$
implies the existence of an FPRAS for~$\Pi$.  
The problem of counting satisfying assignments 
of a Boolean formula is denoted by $\nSAT$. 
Every counting problem in $\numP$ is AP-reducible to $\nSAT$, so $\nSAT$ is said to be
complete for $\numP$ with respect to AP-reductions. It is known that there
is no FPRAS for $\nSAT$ unless $\RP=\NP$.
The problem of counting independent 
sets in a bipartite graph is denoted by $\nBIS$.  The problem $\nBIS$ appears to be
of intermediate complexity:  there is no known FPRAS for $\nBIS$ (and it is generally believed
that none exists) but there is no known AP-reduction from 
$\nSAT$ to $\nBIS$. Indeed, $\nBIS$ is complete with respect to AP-reductions for
a complexity class $\RHPi$ which will be discussed further in Section~\ref{sec:BISeasy}.

We say that a problem $\Pi$ is  
$\nSAT$-hard if there is an AP-reduction from~$\nSAT$ to~$\Pi$,
that it is
$\nSAT$-easy if there is an AP-reduction from~$\Pi$ to~$\nSAT$,
and that it is $\nSAT$-equivalent if both are true.
Note that all of these labels are about the difficulty of \emph{approximately} solving~$\Pi$, not about
the difficulty of exactly solving it.  
Similarly, $\Pi$ is said to be
$\nBIS$-hard if there is an AP-reduction from~$\nBIS$ to~$\Pi$,
$\nBIS$-easy if there is an AP-reduction from~$\Pi$ to~$\nBIS$,
 and
$\nBIS$-equivalent if there are both.

Our main result is 
a trichotomy for the complexity of approximating $\listHcol$.
\begin{theorem}\label{thm:mainunbounded} 
\firstLetH
\begin{enumerate}
\item[(i)] If $H$ is an irreflexive complete bipartite graph or
a reflexive complete graph then $\listHcol$ is in~$\Ptime$.
\item[(ii)] Otherwise, if $H$ is an irreflexive bipartite permutation
graph or a reflexive proper interval graph then $\listHcol$ 
is $\nBIS$-equivalent.
\item[(iii)] Otherwise,  $\listHcol$ is  $\nSAT$-equivalent.
\end{enumerate}
\end{theorem}

\begin{remarks}
\begin{enumerate}
\item 
The assumption that $H$ is connected is made without loss of generality,  
since the complexity of $\listHcol$ is 
determined by the maximum complexity of $\listXcol{H'}$ over all 
connected components $H'$ of~$H$.   To see this,  suppose that $H$ 
has connected components $H_1,\ldots,H_k$.  We can reduce 
$\listXcol{H_i}$ to $\listHcol$ by using the lists to pick out 
the colours in $V(H_i)$.  So hardness results for $H_i$ translate to hardness
results for~$H$.  In the other direction, let $G_1,\ldots,G_m$ be the connected 
components of~$G$.  If we have algorithms for $\listXcol{H_i}$,
for $1\leq i\leq k$, then we can solve $\listXcol{H_i}$ for each instance~$G_j$ to obtain the solution $Z_{i,j}$. Then, to get a solution $Z$ to $\listHcol$
with input~$G$, we combine the solutions via 
$$ Z = \prod_{j=1}^m 
\sum_{i=1}^k Z_{i,j}.$$
\item
Part (ii) of Theorem~\ref{thm:mainunbounded} can be strengthened.  For the graphs~$H$ covered
by this part of the theorem, $\listHcol$ is actually complete for the complexity 
class $\RHPi$.  See Section~\ref{sec:BISeasy}.
\end{enumerate}
\end{remarks}

Theorem~\ref{thm:mainunbounded} also extends to the bounded-degree case.

\begin{theorem}\label{thm:main} 
\firstLetH
\begin{enumerate}
\item[(i)] If $H$ is an irreflexive complete bipartite graph or
a reflexive complete graph then, for all $\Delta$, $\listHcold$ is in~$\Ptime$.
\item[(ii)] Otherwise, if $H$ is an irreflexive bipartite permutation
graph or a reflexive proper interval graph then, for all $\Delta\geq 6$, $\listHcold$ 
is $\nBIS$-equivalent.
\item[(iii)] Otherwise, for all $\Delta\geq 6$, $\listHcold$ is  $\nSAT$-equivalent. Further, if $H$ is reflexive or irreflexive, $\listHcold$ is  $\nSAT$-equivalent for $\Delta\geq 3$.
\end{enumerate}
\end{theorem} 
\begin{remarks}

\begin{enumerate}
\item  The condition $\Delta\geq 6$ is necessary for any hardness result that holds for all graphs $H$.
In particular, there is a graph~$H$ that is not an irreflexive complete bipartite graph or a reflexive complete graph
but for which $\listHcold$ has an FPTAS.
An example is the graph $H=K'_2$ for which   Weitz's self-avoiding walk algorithm \cite{Weitz} gives an FPTAS for
$\listHcold$ for $\Delta \leq 5$.  
\item In general, the lowest value of the degree bound $\Delta$ such that $\listHcold$ is computationally hard depends on the particular graph $H$.  Theorem~\ref{thm:main} leaves open the cases $\Delta=3,4,5$ (partly).
\end{enumerate}
\end{remarks}

Theorems~\ref{thm:mainunbounded}
and~\ref{thm:main}   follow from various constituent results, scattered throughout 
the paper. 

\begin{proof}[Proof of Theorems~\ref{thm:mainunbounded} and~\ref{thm:main}]
Part (i) is trivial.
Part (ii) follows from Lemmas \ref{lem:BIShard} and \ref{lem:BISeasy}.
Part (iii) follows from Lemmas \ref{lem:(ir)reflexive}, \ref{lem:BP}, \ref{lem:BPbounded}, \ref{lem:PI},  and~\ref{lem:PIbounded}.
\end{proof}

The most obvious issue raised by  our theorems is the computational 
complexity of approximately counting $H$-colourings
(in the absence of lists).  This question was extensively studied by
Kelk~\cite{Kelk} and others, and appears much harder to resolve, even when there
are no degree bounds.
It is known~\cite{NoLife} that 
$\Hcol$ is $\nBIS$-hard  for every connected undirected graph~$H$ that
is neither an irreflexive bipartite permutation graph nor a reflexive proper interval graph.
It is not known for which connected~$H$ the problem is $\nBIS$-easy 
and for which it is $\nSAT$-equivalent, and whether one or the other always holds.
In fact, there are 
specific graphs $H$, two of them with as few as four vertices, 
for which the complexity of $\Hcol$ is unresolved.
It is far from clear that a trichotomy should be expected, 
and in fact there may exist an infinite sequences $(H_t)$ of graphs for which 
$\Xcol{H_t}$ is reducible to $\Xcol{H_{t+1}}$ but not vice versa.  Some partial 
results and speculations can be found in~\cite{Kelk}.  

{
\begin{figure}[t]
\begin{center}
\begin{tikzpicture}[xscale=1,yscale=1]
\tikzset{every loop/.style={ in=60, out=120, looseness =7}}

\begin{scope}[shift={(-3,0)}]
\draw (0,0) node[lab] (1) {$1$} ++ (1,0) node[lab] (2) {$2$};
\draw (1) -- (2);
\draw (2) edge [loop above]  (2);
\end{scope}

\begin{scope}
\draw (0,0) node[lab] (1) {$1$} ++ (1,0) node[lab] (2) {$2$} ++ (0.9,0.6) node[lab] (3) {$3$} 
   ++ (0,-1.2) node[lab] (4) {$4$};
\draw (1) -- (2) -- (3);
\draw (2) -- (4);
\draw (2) edge [loop]  (2);
\draw (3) edge [loop above]  (3);
\draw (4) edge [loop below]  (4);
\end{scope}

\begin{scope}[shift={(4,0)}]
\draw (0,0) node[lab] (1) {$1$} ++ (1,0) node[lab] (2) {$2$} ++ (1,0) node[lab] (3) {$3$};
\draw (1) -- (2) -- (3);
\draw (1) edge [loop above]  (1);
\draw (2) edge [loop above]  (2);
\draw (3) edge [loop above]  (3);
\end{scope}

\end{tikzpicture}
\end{center}

\caption{$K_2'$, $2$-wrench and $P_3^*$}
\label{fig:2-wrench}
\end{figure}
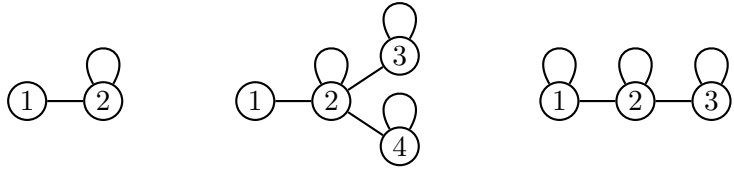
}

As we noted, $\Hcol$ and $\listHcol$ have the same complexity as regards   
exact computation.
However, for approximate computation they are 
different, assuming (as is widely believed) that  
there is no AP-reduction from $\nSAT$ to $\nBIS$.
An example 
is provided by the 2-wrench (see Figure~\ref{fig:2-wrench}).
It is known \cite[Theorem 21]{DGGJ} that  $\Xcol{\text{2-wrench}}$ is $\nBIS$-equivalent, but we know 
from Theorem~\ref{thm:mainunbounded} that the list version $\listXcol{\text{2-wrench}}$ is $\nSAT$-equivalent
since the 2-wrench is neither irreflexive nor reflexive.
One way to see that $\listXcol{\text{2-wrench}}$ is $\nSAT$-equivalent
is  to note 
 that the 2-wrench contains $K_2'$ as an induced subgraph, 
 and that this induced subgraph
 can be ``extracted'' using the list constraints $S_v=\{1,2\}$, for all $v\in V(G)$.
 But $\listXcol{K_2'}$ is already known to be $\nSAT$-equivalent~\cite[Theorem 1]{DGGJ}.
 Indeed, systematic techniques for extracting hard induced subgraphs form the
main theme of the paper.  It is for this reason that the theory of hereditary
graph classes comes into play, just as in~\cite{FHH03}.

Another recent research 
direction, at least in the unbounded-degree case,
 is towards weighted versions of list colouring.
Here, the graph~$H$ is augmented by edge-weights, specifying for each pair 
of colours $i,j$, the cost of assigning 
$i$ and~$j$ to adjacent vertices in~$G$.  
The computational complexity of obtaining approximate 
solutions was studied by Chen, Dyer, Goldberg, 
Jerrum, Lu, McQuillan and Richerby~\cite{ApproxCSP},
and by Goldberg and Jerrum~\cite{PNAS}.  
There is a trichotomy for the case in which the input
has no degree bound, but this is obtained
in a context where {\it individual\/} spins at vertices are weighted 
and not just the interactions between {\it pairs\/} of adjacent spins.  
In this paper we have restricted the class of problems under consideration to
ones having 0,1-weights on interactions, but at the same time we have 
restricted the problem instances to ones having 0,1-weights on individual spins.  
So we have a different tradeoff and the results from the references that we have just discussed
do not carry across, even in the unbounded-degree setting.  Indeed,
towards the end of the paper, in Section~\ref{sec:counterexample},
we give an example to show that Theorem~\ref{thm:mainunbounded} 
is not simply the restriction of earlier results to 0,1-interactions  
(not merely because the proofs differ, but, in a stronger sense, because the
results themselves are different).

Two things are appealing about  our theorems.
First, unlike the weighted classification theorems~\cite{ApproxCSP}, 
 here the truth
is pleasingly simple. The trichotomies for $\listHcol$  and $\listHcold$ have a simple,
natural formulation in terms of hereditary graph classes.
Second, the proofs of the theorems are largely self-contained.
The proofs do not rely on earlier works such as~\cite{ApproxCSP}, which 
require multimorphisms and other deep results from universal algebra.
The proof of Theorem~\ref{thm:mainunbounded} is self-contained apart from some 
 very elementary and well-known starting points, which are 
collected together in Lemma~\ref{lem:zero}.
The proof of Theorem~\ref{thm:main} is similarly self-contained, though it 
additionally relies on recent results~\cite{Sly, GSV}  about approximating the partition function
of the anti-ferromagnetic Ising model on bounded degree graphs (these are also
contained in Lemma~\ref{lem:zero}).

\section{Complexity- and graph-theoretic preliminaries}\label{sec:prelim}
As the complexity of computing exact solutions of $\listHcol$ is well understood,
we focus on the complexity of computing approximations.  
The framework for this 
has already been explained in many papers,
so we provide an informal 
description only here and direct the reader to Dyer, Goldberg, Greenhill 
and Jerrum~\cite{DGGJ} for precise definitions.

The standard notion of efficient approximation algorithm is that of a {\it Fully Polynomial
Randomised Approximation Scheme\/} (or FPRAS).  This is a randomised algorithm 
that is required to produce a solution within 
relative error specified by a tolerance $\varepsilon>0$, in time polynomial in the 
instance size and $\varepsilon^{-1}$.  Evidence for the non-existence of 
an FPRAS for a problem $\Pi$ can be obtained through {\it Approximation-Preserving\/} 
(or AP-) {\it reductions}.  These are randomised polynomial-time Turing reductions that 
preserve (closely enough) the error tolerance.   The set of problems 
that have an FPRAS is closed under AP-reducibility.  

Every problem in $\numP$ is AP-reducible to $\nSAT$, so $\nSAT$ is complete for $\numP$ with respect
to AP-reductions. The same is true of the counting version of any $\NP$-complete decision problem.
It is known that these problems do not have an FPRAS unless $\RP=\NP$. On the other hand, 
using the bisection technique of Valiant and Vazirani 
\cite[Corollary 3.6]{VV}, we know that $\nSAT$ can be approximated (in the FPRAS sense) by a polynomial-time probabilistic
Turing machine equipped with an oracle for the decision problem $\SAT$. 

In the statements and proofs of  our theorems, we refer to two 
hereditary graph classes.  A class of undirected graphs is said to 
be {\it hereditary} if it is closed under taking induced subgraphs.
The classes of bipartite permutation graphs and proper interval graphs have 
been widely studied and many equivalent characterisations of them are 
known.  We are concerned with the excluded subgraph and matrix characterisations.
These characterisations are well known in the area of structural graph theory,
and most can be found in some form in the {\it Information System on Graph Classes 
and their Inclusions\/} (ISGCI) at \texttt{www.graphclasses.org}.  
However, it is not always easy for someone from outside the area
to make the required connections, so we show in an appendix (Section~\ref{sec:app})
how to derive the results we use from the published literature.  Refer to the appendix also for 
proper citations.

The classes of {\it bipartite permutation graphs\/} and {\it proper interval graphs\/}
are defined by certain intersection models, and have 
the following excluded subgraph characterisations.
A graph is a bipartite permutation graph
if and only if it contains none of the following
as an induced subgraph:  $X_3$, $X_2$, $T_2$ or a cycle $C_\ell$ 
of length $\ell$ not equal to four.  (Refer to Figure~\ref{fig:excludedBP}
for specifications of $X_3$, $X_2$ and $T_2$.)   
A graph is a proper interval graph if and only if it contains none of the following
as an induced subgraph:  the claw, the net, $S_3$ or a cycle $C_\ell$ 
of length $\ell$ at least four.  (Refer to Figure~\ref{fig:excludedPI}
for specifications of the claw, the net and $S_3$.)  

These two graph classes also have matrix characterisations.  
Say that a 0,1-matrix 
$A=(A_{i,j}:1\leq i\leq n,1\leq j\leq m)$ 
has 
{\it staircase form\/} if the 1s in each row are contiguous 
and the following condition is satisfied:  letting $\alpha_i=\min\{j:A_{i,j}=1\}$
and $\beta_i=\max\{j:A_{i,j}=1\}$, we require that the sequences $(\alpha_i)$ and 
$(\beta_i)$ are non-decreasing.  It is automatic that the columns share the 
contiguity and monotonicity properties, so the property of having staircase form
is in fact invariant under matrix transposition.   

A graph is a bipartite permutation graph if the rows and columns of its
biadjacency matrix can be (independently) permuted so that the resulting
biadjacency matrix has staircase form.
A reflexive graph is a proper interval graph if there exists a permutation of the vertices so that the resulting adjacency matrix has staircase form.  

As we mentioned in  Section~\ref{sec:intro}, an appealing feature of  our theorems is
that  our proofs are largely self-contained. The only pre-requisites for the proof 
are complexity results classifying some very well-known approximation problems.
These are collected in Lemma~\ref{lem:zero}.
For this, we will use the graph $K'_2$ defined in  Section~\ref{sec:intro} --- see Figure~\ref{fig:2-wrench}.
We will also use the following definitions.
\begin{definition}
Let $P_4$ be the path of length three (with four vertices).   
\end{definition}
\begin{definition}\label{def:Ising}
Let $0<\lambda<1$ be a rational number and let $\Delta$ be a positive integer.
We consider the following problem.
\begin{description}
\item[Name] $\AntiFerroIsingd{\lambda}{\Delta}$.
\item[Instance]  \emph{A graph $G$ of maximum degree at most~$\Delta$.}
\item[Output] \emph{The partition function   of the {\it antiferromagnetic Ising model\/} with parameter~$\lambda$
evaluated on instance~$G$, i.e.,}
$$
Z_\lambda(G)=\sum_{\sigma:V\rightarrow \{\pm1\}} \prod_{\{u,v\}\in E(G)}\lambda^{\delta(\sigma(u),\sigma(v))},
$$
\emph{where $\delta(i,j)$ is 1 if $i=j$ and 0 otherwise.}
\end{description}
\end{definition}

\begin{definition}\label{def:pnsat}
$\pnsat$ is the problem of counting the satisfying assignments of 
a CNF formula 
in which each clause
has at most one negated 
literal 
and at most one unnegated literal.\end{definition}
\begin{remark}
Note that each clause
of an instance of $\pnsat$ 
is either a single literal, or the relation ``implies'' between two variables.
\end{remark}
 
 \begin{lemma}\label{lem:zero}
 The following problems are $\nSAT$-equivalent:
 \begin{itemize}
 \item
 $\Xcold{K'_2}{\Delta}$ for any $\Delta\geq 6$, and
 \item $\AntiFerroIsingd{\lambda}{\Delta}$  
 for any  $\Delta\geq 3$ and $0<\lambda<(\Delta-2)/\Delta$.
 \end{itemize}
 The following problems are $\nBIS$-equivalent: $\Xcold{P_4}{\Delta}$ for $\Delta\geq 6$ and $\pnsat$.
 \end{lemma}
 \begin{proof}
 As we noted in Section~\ref{sec:intro},
 $\Xcold{K'_2}{\Delta}$ 
 is the problem of counting the independent sets of a graph of maximum degree at most~$\Delta$. 
 Without the degree-bound, there is a very elementary proof that this problem is $\nSAT$-equivalent
 (see \cite[Theorem 3]{DGGJ}). This suffices for  the unbounded-degree case.
 For the bounded-degree case,
 the fact that counting independent sets is  $\nSAT$-equivalent for $\Delta\geq 6$ follows from \cite[Theorem 2]{Sly}.  
 In fact, Sly shows 
 in the proof of \cite[Theorem 2]{Sly} that an FPRAS for $\Xcold{K'_2}{\Delta}$ 
 can be used (as an oracle) to provide a polynomial-time
 randomised algorithm for
 the  $\NP$-hard problem \textsc{Max-Cut}.
As is noted in the proof of \cite[Theorem 4]{DGGJ}, this gives an AP-reduction 
from $\nSAT$ to $\Xcold{K'_2}{\Delta}$, using 
the bisection technique of Valiant and Vazirani \cite[Corollary 3.6]{VV}.
 
In the unbounded-degree case, there is an elementary proof
that approximating the partition function of the anti-ferromagnetic Ising model is
$\nSAT$-equivalent. The proof  
  is  an easy reduction from the problem of counting large cuts in a graph, see~\cite[Thm 14]{JS93}.
In the bounded-degree case, 
we have to use more sophisticated results.
  The  fact that $\AntiFerroIsingd{\lambda}{\Delta}$ is $\nSAT$-equivalent for $\Delta\geq 3$ and $0<\lambda<(\Delta-2)/\Delta$  comes from \cite[Theorem 1.2]{GSV}. 
Similar to the paper of Sly mentioned above, the proof of
\cite[Theorem 1.2]{GSV} shows that an FPRAS
for   $\AntiFerroIsingd{\lambda}{\Delta}$ 
can be used to provide a polynomial-time randomised algorithm for \textsc{Max-Cut}
and this can be turned into an AP-reduction from $\nSAT$ to  $\AntiFerroIsingd{\lambda}{\Delta}$ using
the bisection technique of Valiant and Vazirani.

 Note that  $\Xcold{P_4}{\Delta}$ is equivalent to the problem of counting independent sets in bipartite graphs of maximum degree at most~$\Delta$ (this is almost by definition since the end-points of the path can be 
 interpreted as ``in'' the independent set and the other vertices of the path can be interpreted as ``out''). It follows from  \cite[Corollary 3]{CGGGJSV} that the latter problem is $\nBIS$-equivalent for $\Delta\geq 6$, thus yielding that $\Xcold{P_4}{\Delta}$ is $\nBIS$-equivalent for $\Delta\geq 6$.

Finally, the $\nBIS$-equivalence of $\pnsat$ is given in~\cite[Theorem 5]{DGGJ}.
 \end{proof}

\section{$\nSAT$-equivalence}

The aim of this
section is to establish 
the $\nSAT$-equivalence parts of   Theorems~\ref{thm:mainunbounded} and~\ref{thm:main}.

\begin{lemma}\label{lem:(ir)reflexive}
\LetH
If $H$ is 
neither reflexive nor irreflexive then, for all $\Delta\geq 6$, $\listHcold$ is  $\nSAT$-equivalent.
Hence, $\listHcol$ is $\nSAT$-equivalent.
\end{lemma}

\begin{proof}
Let $\Delta\geq 6$. Since $H$ is connected, it must contain $K_2'$ as an induced subgraph.
So $\Xcold{K_2'}{\Delta}$ is AP-reducible to $\listHcold$.  
By Lemma~\ref{lem:zero}, $\Xcold{K_2'}{\Delta}$ is $\nSAT$-equivalent.
\end{proof} 
\begin{remark}
We can see already that there is a complexity gap between $\listHcol$ and $\Hcol$.
The smallest witness to this gap is the 2-wrench mentioned in  Section~\ref{sec:intro} (Figure~\ref{fig:2-wrench}).  The problem
$\Xcol{\text{2-wrench}}$ is
$\nBIS$-equivalent~\cite[Theorem 21]{DGGJ}, whereas the problem $\listXcol{\text{2-wrench}}$ is 
$\nSAT$-equivalent
by 
Lemma~\ref{lem:(ir)reflexive}.  
The point is that a graph~$H$ for which 
$\Hcol$ is $\nBIS$-easy may contain an induced subgraph $H'$ for
which $\Xcol{H'}$ is  $\nSAT$-equivalent.   In other words, the class of graphs~$H$
such that $\Hcol$ is $\nBIS$-easy is not hereditary.
Identifying  $\nSAT$-equivalent subgraphs is the
main analytical tool in this section.  For this we use results in structural 
graph theory.
\end{remark}

 The gadgets that we use in our reductions in Sections~\ref{sec:A} and~\ref{sec:B} 
are of a particularly simple kind, 
namely paths.\footnote{We were 
also able to make use of path gadgets in~\cite{PNAS}, though, as noted (see Section~\ref{sec:intro}) the results unfortunately do not carry over to our setting. Here the use of structural graph theory makes the
discovery of such gadgets pleasingly straightforward.}  Let the vertex set of the 
$L$-vertex
path be $\{1,2,\ldots,L\}$,
where the vertices are numbered according to their position on the path. 
The end vertices $1$ and~$L$ are {\it terminals}, which make connections 
with the rest of the construction.  
For each vertex $1\leq k\leq L$ there is a set of allowed colours~$S_k$.
We can describe a gadget by specifying $L$ and specifying the sets 
$(S_1,S_2,\ldots,S_L)$.  In our application, each set $S_i$ has cardinality 2,
and $S_1=S_L$.

Fix a connected graph~$H$, possibly with loops.  Our strategy for proving 
that $\listHcold$ is $\nSAT$-equivalent 
is to find a gadget 
$(\{i_1,j_1\}, \{i_2,j_2\}, \ldots, \{i_L,j_L\})$
such that 
\begin{enumerate}[(i)]
\item the sequence $(i_1,\ldots,i_L)$ is a path in $H$,
and likewise $(j_1,\ldots,j_L)$;  
\item it is never the case that both $\{i_{k},j_{k+1}\}\in E(H)$ 
and $\{j_k,i_{k+1}\}\in E(H)$; and
\item $i_1=j_L$ and $j_1=i_L$.
\end{enumerate}
If we achieve these conditions then, as we shall see, the colours
at the terminals will be negatively correlated, and from there
we will be able to encode instances of  $\AntiFerroIsingd{\lambda}{\Delta}$ for some integer $\Delta\geq 3$ and $\lambda\in(0,\frac{\Delta-2}{\Delta})$, 
and this is $\nSAT$-equivalent (Lemma~\ref{lem:zero}).  Note that although 
the ordering of elements within the sets~$S_i$ is irrelevant to the 
workings of the gadget, we write the pairs in a specific order
to bring out the path structure  that we have just described.

Fix $H$ and let $A=A_H$ be the adjacency matrix of~$H$. 
Denote by $\submat{A}{i}{j}{i'}{j'}$ the $2\times2$ submatrix of~$A$ 
indexed by rows $i$ and $j$ and columns $i'$ and $j'$.
We regard the indices in the notation $\submat{A}{i}{j}{i'}{j'}$ as ordered;
thus the first row of this $2\times2$ matrix comes from row $i$ of~$A$
and the second from row~$j$. 

Given a gadget, i.e., sequence 
$(\{i_1,j_1\}, \{i_2,j_2\}, \ldots, \{i_L,j_L\})$, consider the 
product of $2\times2$ submatrices of $A$:
\begin{equation}\label{eq:D'defn}
D'=\submat{A}{i_1}{j_1}{i_2}{j_2} \submat{A}{i_2}{j_2}{i_3}{j_3}
\cdots 
\submat{A}{i_{L-1}}{j_{L-1}}{i_L}{j_L}.
\end{equation}
If conditions (i)--(iii) for gadget construction are satisfied then 
each of the $2\times2$ matrices in the product has  1s on the diagonal;
also, all of them have at least one off-diagonal entry that is 0.  
Thus, each matrix has determinant 1, from which it follows that $\det D'=1$.

Now consider the matrix $D$ that is obtained by swapping the 
two columns of~$D'$.  This swap rectifies the ``twist'' that 
occurs in the passage from $(i_1,j_1)$ to $(i_L,j_L)=(j_1,i_1)$, 
but it also flips the sign of the determinant, leaving $\det D=-1$.  
Let $r=i_1=j_L$ and $s=j_1=i_L$.  The matrix~$D$ can
be interpreted as giving the number of list $H$-colourings of the gadget 
when the $k$'th vertex of the gadget (for $k\in \{1,\ldots,L\}$ is 
assigned the list 
$\{i_k,j_k\}$, so the
terminals are restricted
to colours in $\{r,s\}$.
Thus 
\begin{itemize}
\item the entry in
the first row and column of~$D$ is the number of colourings with both terminals 
receiving~colour~$r$, 
\item the entry in the first row and second column is the 
number of colourings with terminal~$1$ receiving colour~$r$ and terminal~$L$ 
receiving colour~$s$,  
\item 
the entry in the second row and first column is the number of colourings with terminal~$1$
receiving colour~$s$ and terminal~$L$ receiving colour~$r$, and finally
\item   the entry in the
second row and second column is the number of colourings with both terminals receiving colour~$s$.
\end{itemize}
We call $D=D(\Gamma)$ the {\it 
interaction matrix\/} associated with the gadget~$\Gamma$. 
Since $\det D<0$ the gadget provides a negative
correlation between the colours at the terminals, 
which, as we will see, will allow
a reduction
from $\AntiFerroIsingd{\lambda}{\Delta}$. 

In Sections \ref{sec:A} and \ref{sec:B}, we first demonstrate how to apply the technique to get the  $\nSAT$-equivalences in the unbounded-degree case. For these arguments, we intentionally keep the construction of the gadgets as simple as possible. While this would also lead to a value of $\Delta$ such that $\listHcold$ is $\nSAT$-equivalent, the value of $\Delta$ would be much larger than $6$. Thus, our remaining task will be to show how to refine the constructions to obtain the bounded-degree results of Theorem~\ref{thm:main}. 

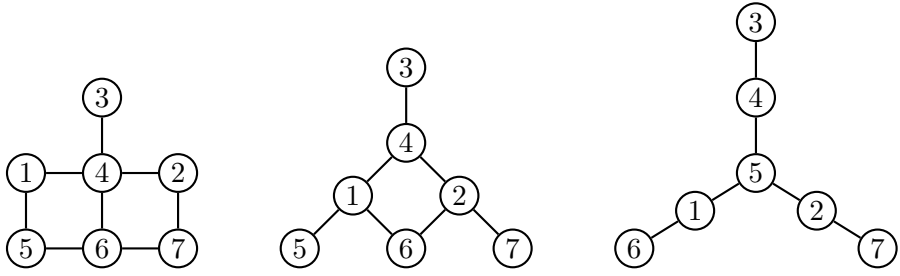
\begin{figure}[t]
\begin{center}
\begin{tikzpicture}[xscale=1,yscale=1]

\begin{scope}[shift={(3,0)}]
\draw (0,0) node[lab] (1) {$6$} ++ (0.8,0.5) node[lab] (2) {$1$} ++ (0.8,0.5) node[lab] (7) {$5$} 
   ++ (0,1) node[lab] (4) {$4$} ++ (0,1) node[lab] (3) {$3$} ++ (0.8,-2.5) node[lab] (6) {$2$} 
   ++ (0.8,-0.5) node[lab] (5) {$7$};
\draw (1) -- (2) -- (7) -- (4) -- (3);
\draw (7) -- (6) -- (5);
\end{scope}

\begin{scope}[shift={(-5,0)}]
\draw (0,0) node[lab] (2) {$5$} ++ (1,0) node[lab] (1) {$6$} ++ (1,0) node[lab] (4) {$7$} 
   ++ (0,1) node[lab] (5) {$2$} ++ (-1,0) node[lab] (6) {$4$} ++ (-1,0) node[lab] (3) {$1$} 
   ++ (1,1) node[lab] (7) {$3$};
\draw (1) -- (2) -- (3) -- (6) -- (5) -- (4) -- (1) -- (6) -- (7);
\end{scope}

\draw (0,0) node[lab] (1) {$6$} ++ (0.7,0.7) node[lab] (4) {$2$} ++ (0.7,-0.7) node[lab] (5) {$7$} 
   ++ (-1.4,1.4) node[lab] (7) {$4$} ++ (0,1) node[lab] (6) {$3$} ++ (-0.7,-1.7) node[lab] (2) {$1$} 
   ++ (-0.7,-0.7) node[lab] (3) {$5$};
\draw (2) -- (1) -- (4) -- (5);
\draw (4) -- (7) -- (6);
\draw (7) -- (2) -- (3);

\end{tikzpicture}
\end{center}

\caption{$X_3$, $X_2$ and $T_2$}
\label{fig:excludedBP}
\end{figure}

\subsection{Irreflexive graphs that are not bipartite permutation graphs}
\label{sec:A}

\begin{lemma}\label{lem:BP}
\LetH
If $H$ is irreflexive but it is not a bipartite permutation
graph, then $\listHcol$ is  $\nSAT$-equivalent.
\end{lemma}

\begin{proof}
Graphs that are not bipartite permutation graphs
contain one of the following as an induced subgraph:
$X_3$, $X_2$, $T_2$, or a cycle of length other than~4.  (Refer to 
Figure~\ref{fig:excludedBP}.)  We just have to show that $\listHcol$
is  $\nSAT$-equivalent when $H$ is any of these. 

We consider the case~$X_3$ in detail, and the others more swiftly, as
they all follow the same general pattern.  The gadget in this case is 
$$
\Gamma=\big(\{1,2\},\{4,7\},\{3,6\},\{4,5\},\{2,1\}\big).
$$
Conditions (i) and (iii) for gadget construction are immediately satisfied,
while condition (ii) is easy to check.
Explicit calculation using (\ref{eq:D'defn}) yields
$$
D'=\submat{A}{1}{2}{4}{7}\submat{A}{4}{7}{3}{6}\submat{A}{3}{6}{4}{5}\submat{A}{4}{5}{2}{1}
=(\begin{smallmatrix}1&0\\1&1\end{smallmatrix})
(\begin{smallmatrix}1&1\\0&1\end{smallmatrix})
(\begin{smallmatrix}1&0\\1&1\end{smallmatrix})
(\begin{smallmatrix}1&1\\0&1\end{smallmatrix})=
(\begin{smallmatrix}2&3\\3&5\end{smallmatrix}).
$$
Swapping the columns of $D'$ yields the interaction matrix
$D=(\begin{smallmatrix}3&2\\5&3\end{smallmatrix})$.
As we explained earlier, $\det D=-1$.  
Obtaining a matrix~$D$ with negative determinant
is moving in the right direction, but in 
order to encode antiferromagnetic Ising we ideally want the matrix $D=(D_{i,j})$
to also satisfy $D_{1,1}=D_{2,2}$ and $D_{1,2}=D_{2,1}$.

In the case of~$X_2$, we have already $D_{1,1}=D_{2,2}$, which makes 
the task easier.  But we are not always in this favourable situation, so
we introduce a technique that works in general for 
all of the graphs that we consider.

Observe that the graph $X_3$ has an automorphism of order two,
$\pi=(1,2)(5,7)$, that transposes vertices 1 and~2,
which are the terminals of the gadget~$\Gamma$.
Consider the gadget obtained from $\Gamma$ by letting $\pi$
act on the colour sets, namely
\begin{align*}
\Gamma^\pi&=\big(\{\pi(1),\pi(2)\},\{\pi(4),\pi(7)\},
\{\pi(3),\pi(6)\},\{\pi(4),\pi(5)\},\{\pi(2),\pi(1)\}\big)\\
&=\big(\{2,1\},\{4,5\},\{3,6\},\{4,7\},\{1,2\}\big).
\end{align*}
The interaction matrix 
$D^\pi=(\begin{smallmatrix}3&5\\2&3\end{smallmatrix})$ 
corresponding to $\Gamma^\pi$ is the same as~$D$, 
except that the rows and columns are swapped.  
Placing $\Gamma$ and $\Gamma^\pi$ in parallel,
identifying the terminals, 
yields a composite gadget $\Gamma^*$    whose interaction matrix is 
\begin{equation}\label{eq:X3}
D^*=\begin{pmatrix}D_{1,1}D_{2,2}&D_{1,2}D_{2,1}\\D_{2,1}D_{1,2}&D_{2,2}D_{1,1}\end{pmatrix}
=\begin{pmatrix}9&10\\10&9\end{pmatrix}.
\end{equation}
 Clearly the same construction will work for any 
graph $H$ with an automorphism swapping the terminals of~$\Gamma$,
provided $D>0$. Note that the gadget $\Gamma^*$ has maximum degree 2 (this observation will be important in the upcoming Lemma~\ref{lem:BPbounded}). 
Also, in general, $\det D^* =D_{1,1}^2D_{2,2}^2-D_{1,2}^2D_{2,1}^2=
(D_{1,1}D_{2,2}+D_{1,2}D_{2,1})\det D<0$.  
So we have an AP-reduction from 
$\AntiFerroIsing{\lambda}$ with 
$\lambda={D_{1,1} D_{2,2}/ (D_{1,2} D_{2,1})}$ 
to $\listXcol{H}$:    given an instance~$G$ of 
$\AntiFerroIsing{\lambda}$, simply replace each
edge $\{u,v\}$ of~$G$ with a copy of the gadget~$\Gamma^*$, identifying the 
two terminals of~$\Gamma^*$
with the vertices $u$ and~$v$, respectively.
(Since $\Gamma^*$ is symmetric, it does not matter which is~$u$ and which is~$v$.)
The problem $\AntiFerroIsing{\lambda}$ is $\nSAT$-equivalent by Lemma~\ref{lem:zero}.
In the case $H=X_3$, we have    $\lambda=\frac9{10}$. 

Now we present in less detail the gadgets for  $X_2$,
$T_2$, odd cycles, and even cycles of length
at least~$6$.  For $X_2$,
the gadget is
$$
\big(\{1,2\},\{4,7\},\{3,2\},\{4,6\},\{3,1\},\{4,5\},\{2,1\}\big),
$$
with 
$$
D'=(\begin{smallmatrix}1&0\\1&1\end{smallmatrix})
(\begin{smallmatrix}1&1\\0&1\end{smallmatrix})
(\begin{smallmatrix}1&0\\1&1\end{smallmatrix})
(\begin{smallmatrix}1&1\\0&1\end{smallmatrix})
(\begin{smallmatrix}1&0\\1&1\end{smallmatrix})
(\begin{smallmatrix}1&1\\0&1\end{smallmatrix})=
(\begin{smallmatrix}5&8\\8&13\end{smallmatrix}).
$$
The interaction matrix is $D=(\begin{smallmatrix}8&5\\13&8\end{smallmatrix})$. The graph $X_2$ has an automorphism transposing $1$ and~$2$, yielding the symmetrised interaction matrix
\begin{equation*}
D^*=\begin{pmatrix}64&65\\65&64\end{pmatrix}.
\end{equation*}
The remaining part of the analysis can be completed exactly as before. So that we don't need to repeat this observation in future, let us note at this point that 
all the graphs $H$ we consider in this proof and the next
have an automorphism of order two transposing the 
two distinguished terminal colours.

For $T_2$ the gadget is 
$$
\big(\{1,2\},\{5,7\},\{4,2\},\{3,5\},\{4,1\},\{5,6\},\{2,1\}\big),
$$
with 
$$
D'=(\begin{smallmatrix}1&0\\1&1\end{smallmatrix})
(\begin{smallmatrix}1&1\\0&1\end{smallmatrix})
(\begin{smallmatrix}1&1\\0&1\end{smallmatrix})
(\begin{smallmatrix}1&0\\1&1\end{smallmatrix})
(\begin{smallmatrix}1&0\\1&1\end{smallmatrix})
(\begin{smallmatrix}1&1\\0&1\end{smallmatrix})=
(\begin{smallmatrix}5&7\\7&10\end{smallmatrix}).
$$
The interaction matrix is $D=(\begin{smallmatrix}7&5\\10&7\end{smallmatrix})$,  yielding the symmetrised interaction matrix
\begin{equation*}
D^*=\begin{pmatrix}49&50\\50&49\end{pmatrix}.
\end{equation*}

We will conclude by presenting the gadgets for odd cycles and for even cycles of length
at least~$6$. The reason for doing so is that we will use the gadgets in the upcoming Lemma~\ref{lem:BPbounded} to  obtain the result for the bounded-degree case (and
to present easy, self-contained proofs). 
For the unbounded-degree case, the remainder of the argument could be omitted since 
the result follows easily from the fact that the decision problem is $\NP$-hard
in these cases \cite[Theorem 3.1]{FHH}.

For a cycle of even length $q\geq6$ the gadget is
$$
\big(\{1,3\},\{2,4\},\{1,5\},\ldots\{1,q-1\},\{2,q\},\{3,1\}\big).
$$
Note that, for convenience, the terminal colours in this case are $1$ and~$3$,
rather than $1$ and~$2$, as elsewhere.  
To clarify the construction, we are setting $L=q-1$, and the intention is that 
the path $i_1,\ldots,i_L$ oscillates
between $1$ and~$2$, before moving at the last step to~$3$,
while the path $j_1,\ldots,j_L$ cycles clockwise from $3$ to~$1$.
We have
$$
D'=(\begin{smallmatrix}1&0\\1&1\end{smallmatrix})
(\begin{smallmatrix}1&0\\0&1\end{smallmatrix})
(\begin{smallmatrix}1&0\\0&1\end{smallmatrix})\cdots
(\begin{smallmatrix}1&0\\0&1\end{smallmatrix})
(\begin{smallmatrix}1&1\\0&1\end{smallmatrix})
(\begin{smallmatrix}1&1\\0&1\end{smallmatrix})=
(\begin{smallmatrix}1&2\\1&3\end{smallmatrix}).
$$
Note that this construction fails for $q=4$!
The interaction matrix is 
$D=(\begin{smallmatrix}2&1\\3&1\end{smallmatrix})$ and its symmetrised version is
\begin{equation*}
D^*=\begin{pmatrix}2&3\\3&2\end{pmatrix}.
\end{equation*}.

Finally, for a cycle of odd length~$q$ the gadget is
$$
\big(\{1,2\},\{2,3\},\ldots,\{1,q-1\},\{2,q\},\{1,q-1\},\ldots,\{2,3\},\{1,2\},\{2,1\}\big).
$$
To clarify the construction, we are setting $L=2q-2$, and the intention is that
the path $i_1,\ldots,i_L$ oscillates between $1$ and~$2$, 
while the path $j_1,\ldots,j_L$ cycles clockwise from $2$ to~$q$ and then 
anticlockwise back to $2$ and then on to~$1$.
We have
$$
D'=(\begin{smallmatrix}1&0\\0&1\end{smallmatrix})\cdots
(\begin{smallmatrix}1&0\\0&1\end{smallmatrix})
(\begin{smallmatrix}1&1\\0&1\end{smallmatrix})
(\begin{smallmatrix}1&0\\1&1\end{smallmatrix})
(\begin{smallmatrix}1&0\\0&1\end{smallmatrix})\cdots
(\begin{smallmatrix}1&0\\0&1\end{smallmatrix})
(\begin{smallmatrix}1&0\\0&1\end{smallmatrix})=
(\begin{smallmatrix}2&1\\1&1\end{smallmatrix}).
$$
Note that this construction works even for $q=3$. The corresponding interaction matrix is $D=(\begin{smallmatrix}1&2\\1&1\end{smallmatrix})$, and its symmetrised version is
\begin{equation*}
D^*=\begin{pmatrix}1&2\\2&1\end{pmatrix}.
\end{equation*}

In all cases, we obtain a reduction from $\AntiFerroIsing{\lambda}(\Delta)$, completing the proof.
\end{proof}

As noted earlier, we need to consider slightly more complicated gadgets to get the bounded degree results of Theorem~\ref{thm:main}. Roughly, the idea is to implement thickenings of the gadgets using carefully chosen list colourings to keep the degree of the gadget small. This is achieved in the next lemma.
\begin{lemma}\label{lem:BPbounded}
\LetH
If $H$ is irreflexive but it is not a bipartite permutation
graph, then for all $\Delta\geq 3$, $\listHcold$ is  $\nSAT$-equivalent. 
\end{lemma}

\begin{proof}
As in the proof of Lemma~\ref{lem:BP}, it suffices to show that, for all $\Delta\geq 3$, $\listHcold$
is  $\nSAT$-equivalent when $H$ is any of $X_3$, $X_2$, $T_2$, or a cycle of length other than~4.  (Refer to 
Figure~\ref{fig:excludedBP}.) 

Again, we consider the case~$X_3$ in detail, and the other more swiftly, as
they all follow the same general pattern.  In Lemma~\ref{lem:BP} we  constructed a gadget $\Gamma^{*}$ whose terminals have 
the allowed set of colours $\{1,2\}$  and 
whose interaction matrix is given by $D^{*}=(\begin{smallmatrix}9&10\\10&9\end{smallmatrix})$ (see Equation~\eqref{eq:X3}). The gadget $\Gamma^*$ was obtained by placing two paths in parallel and hence all of its vertices have degree 2. Note that all of the gadgets $\Gamma^*$ in the proof of Lemma~\ref{lem:BP} had the same property, so we will not repeat this observation later on.

Let $t\geq 0$ be an integer. We will denote by $D^{*}_t$ the matrix whose entries are obtained from $D^{*}$ by raising the entries of $D^{*}$ to the power $2^t$. Thus, in the case of $X_3$, $D^{*}_t=\big(\begin{smallmatrix}9^{2^t}&10^{2^t}\\10^{2^t}&9^{2^t}\end{smallmatrix}\big)$. We will construct inductively a gadget $\Gamma^*_t$ with the following properties:
\begin{enumerate}[(i)]
\item The allowed colours of the terminals of $\Gamma^*_t$ will be $\{1,2\}$ for odd $t$ and $\{5,7\}$ for even $t$.
\item  The two terminals of $\Gamma^*_t$ will each have degree 1, and all other vertices 
of~$\Gamma^*_t$ will
have degree at most~$3$.
\item  The interaction matrix of $\Gamma^*_t$ will be $D^{*}_t$.
\end{enumerate}
We first do the base case $t=0$. Let $u,v$ be the terminals of $\Gamma^*$ and recall that their allowed sets of colours is $\{1,2\}$. The gadget $\Gamma^*_0$ (i.e., $t=0$) is obtained from $\Gamma^*$ by adding two new vertices $u_0,v_0$ and adding the edges  $(u_0,u),\, (v_0,v)$.  For all vertices that were initially in $\Gamma^*$, we keep their sets of allowed colours the same, and for $u_0,v_0$ we restrict their allowed colours to the set $\{5,7\}$.  Finally, we set the terminals of $\Gamma^*_0$ to be $u_0,v_0$. It is  easy to see  that $\Gamma^*_0$ satisfies properties (i) and (ii) for $t=0$. To find the interaction matrix of $\Gamma^*_0$, note that colour $1$ is adjacent to colour $5$ in $X_3$ but not colour $7$, and similarly colour $2$ is adjacent to colour $7$ in $X_3$ but not colour $5$. It follows that 
\[D(\Gamma^*_0)=\begin{pmatrix}1&0\\0&1\end{pmatrix}\begin{pmatrix}9&10\\10&9\end{pmatrix}\begin{pmatrix}1&0\\0&1\end{pmatrix}=D^*_0,\]
proving that $\Gamma^*_0$ satisfies all properties (i)---(iii), as desired.

We now carry out the induction step, so assume that for an integer $t\geq 0$, we have $\Gamma^*_t$ which satisfies all properties (i)---(iii). Let $u_t,v_t$ be the terminals of $\Gamma^*_t$. Take two copies of $\Gamma^*_t$ and place them in parallel, identifying the two copies of $u_t$ into one vertex and similarly for the two copies of $v_t$. Also, add two new vertices $u_{t+1},v_{t+1}$ and add the edges  $(u_{t+1},u_t),\, (v_{t+1},v_t)$. For all vertices that belonged to one of the copies of $\Gamma^*_t$, we keep their sets of allowed colours the same, and for $u_{t+1},v_{t+1}$ we restrict their allowed colours to the set $\{1,2\}$ if $t$ is even and to $\{5,7\}$ if $t$ is odd. The final graph is the gadget $\Gamma^*_{t+1}$ with terminals $u_{t+1},v_{t+1}$. It is immediate that $\Gamma^*_{t+1}$ satisfies property (i) and, using the fact that $\Gamma^*_{t}$ satisfies property (ii), it is simple to check that $\Gamma^*_{t+1}$ satisfies property (ii) as well. To see property (iii), note by construction that the set of allowed colours of $u_{t+1}$ and $u_t$ are either $\{1,2\}$ and $\{5,7\}$ for even $t$ or $\{5,7\}$ and $\{1,2\}$ for odd $t$, respectively. Taking also into consideration the 2-thickening of the gadget $\Gamma^*$, we obtain
\[D(\Gamma^*_{t+1})=\begin{pmatrix}1&0\\0&1\end{pmatrix}\begin{pmatrix}(9^{2^t})^2&(10^{2^t})^2\\(10^{2^t})^2&(9^{2^t})^2\end{pmatrix}\begin{pmatrix}1&0\\0&1\end{pmatrix}=D^*_{t+1},\]
thus proving property (iii) as well. This completes the induction and hence the construction of the gadgets $\Gamma^{*}_t$. 

Let $\lambda_t=9^{2^t}/10^{2^t}$ for an arbitrary integer $t\geq 0$. For every positive integer~$\Delta\geq 3$,
we now have an AP-reduction from  
$\AntiFerroIsingd{{\lambda_t}}{\Delta}$ to $\listXcold{H}{\Delta}$:    given an instance~$G$ of 
$\AntiFerroIsingd{{\lambda_t}}{\Delta}$, simply replace each
edge $\{u,v\}$ of~$G$ with a copy of the gadget~$\Gamma^*_t$, identifying the 
two terminals of~$\Gamma^*_t$ with the vertices $u$ and~$v$, respectively.  Using property (ii) of the gadget $\Gamma^*_t$,  the resulting instance of $\listXcol{X_3}$ has maximum degree at most $\Delta$. We can clearly find a positive integer $t$ such that $\lambda_t<(\Delta-2)/\Delta$, so that $\AntiFerroIsingd{{\lambda_t}}{\Delta}$ is $\nSAT$-equivalent by Lemma~\ref{lem:zero}. Thus, we obtain that $\listXcold{X_3}{\Delta}$ is $\nSAT$-equivalent for all $\Delta\geq 3$, as claimed.

To expedite the proof of the remaining cases for the graph $H$, note that the part of the argument which was crucial in the construction of the gadgets $\Gamma^*_t$ from $\Gamma^*$ was the following condition:
\begin{equation}\label{eq:condH}
\begin{array}{l}
\mbox{for the allowed colours $r,s$ for the terminals of $\Gamma^{*}$, there exist colours $r',s'$}\\
\mbox{such that $r$ is adjacent to $r'$ but not $s'$, and $s$ is adjacent to $s'$ but not $r'$.} 
\end{array}
\end{equation}
Note that 
in some circumstances it will suffice to have  $\{r,s\}=\{r',s'\}$. For example, $r'=s$ and $s'=r$
is a solution  if $(r,s)$ is an edge of $H$ and there are no self-loops on~$r$ or~$s$.   If condition \eqref{eq:condH} holds for a graph $H$ and the respective gadget $\Gamma^*$ in Lemma~\ref{lem:BP}, the above argument  yields gadgets $\Gamma^*_t$ satisfying properties (i)---(iii). We next check that this is the case for all the graphs $H$ under consideration with the single exception of the cycle of length 3.

For $X_2$, the set of allowed colours for the terminals of $\Gamma^*$ in Lemma~\ref{lem:BP} was $\{1,2\}$ and hence the colours in $\{5,7\}$ establish condition \eqref{eq:condH}.

For $T_2$, the set of allowed colours for the terminals of $\Gamma^*$ in Lemma~\ref{lem:BP} was $\{1,2\}$ and hence the colours in $\{6,7\}$ establish condition \eqref{eq:condH}.

For a cycle of even length $q \geq  6$, the set of allowed colours for the terminals of $\Gamma^*$ in Lemma~\ref{lem:BP} was $\{1,3\}$ and hence the colours in $\{q,4\}$ establish condition \eqref{eq:condH}.

For a cycle of odd length $q\geq 3$, the set of allowed colours for the terminals of $\Gamma^*$ in Lemma~\ref{lem:BP} was $\{1,2\}$ and hence the colours in $\{2,1\}$ establish condition \eqref{eq:condH}. (Note that there are no self-loops on $1,2$.)

Thus, when $H$ is any of $X_3$, $X_2$, $T_2$, or a cycle of length other than 4, we obtain, for all $\Delta\geq 3$, for some $0<\lambda<(\Delta-2)/\Delta$, an AP-reduction from  $\AntiFerroIsingd{{\lambda}}{\Delta}$  to $\listXcold{H}{\Delta}$. By Lemma~\ref{lem:zero}, $\AntiFerroIsingd{{\lambda}}{\Delta}$ is $\nSAT$-equivalent, so $\listXcold{H}{\Delta}$ is also $\nSAT$-equivalent.
\end{proof}

\subsection{Reflexive graphs that are not proper interval graphs}
\label{sec:B}

\begin{figure}[t]
\begin{center}
\begin{tikzpicture}[xscale=1,yscale=1]

\begin{scope}[shift={(-4,0)}]
\draw (0,0) node[lab] (1) {$1$} ++ (0.8,0.5) node[lab] (4) {$4$} ++ (0,1) node[lab] (3) {$3$} 
   ++ (0.8,-1.5) node[lab] (2) {$2$} ;
\draw (1) -- (4) -- (3);
\draw (4) -- (2);
\end{scope}

\draw (0,0) node[lab] (5) {$5$} ++ (0.7,0.7) node[lab] (1) {$1$} ++ (0.5,0.8) node[lab] (4) {$4$} 
   ++ (0,1) node[lab] (3) {$3$} ++ (0.5,-1.8) node[lab] (2) {$2$} ++ (0.7,-0.7) node[lab] (6) {$6$} ;
\draw (5) -- (1) -- (4) -- (2) -- (6);
\draw (3) -- (4);
\draw (1) -- (2);

\begin{scope}[shift={(4.8,0)}]
\draw (0,0) node[lab] (4) {$4$} ++ (0.5,0.8) node[lab] (1) {$1$} ++ (0.5,0.8) node[lab] (3) {$3$} 
   ++ (0.5,-0.8) node[lab] (2) {$2$} ++ (0.5,-0.8) node[lab] (6) {$6$} ++ (-1,0) node[lab] (5) {$5$}; 
\draw (4) -- (1) -- (3) -- (2) -- (6) -- (5) -- (4);
\draw (1) -- (2) -- (5) -- (1) ;

\end{scope}
\end{tikzpicture}
\end{center}

\caption{The claw, the net and $S_3$}
\label{fig:excludedPI}
\end{figure}
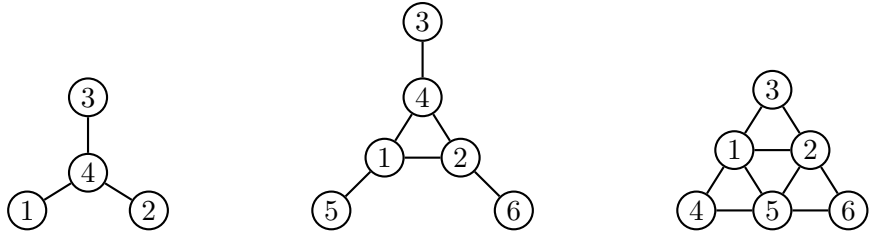

\begin{lemma}\label{lem:PI}
\LetH
If $H$ is a   reflexive graph that is not a proper interval 
graph, then $\listHcol$ is  $\nSAT$-equivalent.
\end{lemma}

\begin{proof}
The line of argument is  the same as in the proof of Lemma~\ref{lem:BP}.
Graphs that are not proper interval graphs
contain one of the following as an induced subgraph:
the claw, the net, $S_3$, or a cycle of length at least four.  
(Refer to 
Figure~\ref{fig:excludedPI} but note that loops are omitted.)  
We just have to show that $\listHcol$ is  $\nSAT$-equivalent when $H$ 
is any of these. 

For the claw, the gadget we use is 
$$
\big(\{1,2\},\{4,2\},\{3,4\},\{4,1\},\{2,1\}\big),
$$
with 
$$
D'=
(\begin{smallmatrix}1&0\\1&1\end{smallmatrix})
(\begin{smallmatrix}1&1\\0&1\end{smallmatrix})
(\begin{smallmatrix}1&0\\1&1\end{smallmatrix})
(\begin{smallmatrix}1&1\\0&1\end{smallmatrix})
=
(\begin{smallmatrix}2&3\\3&5\end{smallmatrix}),
$$
and interaction matrix $D=(\begin{smallmatrix}3&2\\5&3\end{smallmatrix})$.
The claw has an automorphism transposing vertices $1$ and $2$ (as do the
other graphs graphs $H$ that we meet in this proof) so we can complete the
reduction from the partition function of the antiferromagnetic Ising model
  as in the proof of Lemma~\ref{lem:BP}.

For the net, the gadget is
$$
\big(\{1,2\},\{4,6\},\{3,2\},\{3,1\},\{4,5\},\{2,1\}\big),
$$
with 
$$
D'=(\begin{smallmatrix}1&0\\1&1\end{smallmatrix})
(\begin{smallmatrix}1&1\\0&1\end{smallmatrix})
(\begin{smallmatrix}1&0\\0&1\end{smallmatrix})
(\begin{smallmatrix}1&0\\1&1\end{smallmatrix})
(\begin{smallmatrix}1&1\\0&1\end{smallmatrix})
=
(\begin{smallmatrix}2&3\\3&5\end{smallmatrix}).
$$
and interaction matrix $D=(\begin{smallmatrix}3&2\\5&3\end{smallmatrix})$.

For $S_3$ it is 
$$
\big(\{1,2\},\{3,6\},\{3,5\},\{3,4\},\{2,1\}\big),
$$
with 
$$
D'=(\begin{smallmatrix}1&0\\1&1\end{smallmatrix})
(\begin{smallmatrix}1&0\\0&1\end{smallmatrix})
(\begin{smallmatrix}1&0\\0&1\end{smallmatrix})
(\begin{smallmatrix}1&1\\0&1\end{smallmatrix})
=
(\begin{smallmatrix}1&1\\1&2\end{smallmatrix}).
$$
and interaction matrix $D=(\begin{smallmatrix}1&1\\2&1\end{smallmatrix})$.

Finally, for the cycle of length $q\geq4$ it is
$$
\big(\{1,2\},\{1,3\},\ldots,\{1,q-1\},\{1,q\},\{2,1\}\big).
$$
Here $L=q$, and the path $i_1,\ldots,i_L$ cycles round the loop at vertex~$1$
and moves to $2$ at the last step, 
while the path $j_1,\ldots,j_L$ cycles clockwise from vertex $2$ to vertex~$1$.
We have
$$
D'=(\begin{smallmatrix}1&0\\1&1\end{smallmatrix})
(\begin{smallmatrix}1&0\\0&1\end{smallmatrix})
(\begin{smallmatrix}1&0\\0&1\end{smallmatrix})\cdots
(\begin{smallmatrix}1&0\\0&1\end{smallmatrix})
(\begin{smallmatrix}1&1\\0&1\end{smallmatrix})
(\begin{smallmatrix}1&1\\0&1\end{smallmatrix})=
(\begin{smallmatrix}1&2\\1&3\end{smallmatrix}).
$$
and the interaction matrix is $D=(\begin{smallmatrix}2&1\\3&1\end{smallmatrix})$.
This completes the analysis of the excluded subgraphs and the proof. 
\end{proof}

\begin{lemma}\label{lem:PIbounded}
\LetH
If $H$ is a   reflexive graph that is not a proper interval 
graph, then, for $\Delta\geq 3$, $\listHcold$ is  $\nSAT$-equivalent.
\end{lemma}

\begin{proof}
Similarly to our derivation of Lemma~\ref{lem:BPbounded} from Lemma~\ref{lem:BP}, to derive the lemma from Lemma~\ref{lem:PI}, we only need to check that condition \eqref{eq:condH} holds when $H$  is the claw, the net, $S_3$, or a cycle of length at least four (we iterate the condition for convenience):
\begin{equation*}\tag{\ref{eq:condH}}
\begin{array}{l}
\mbox{for the allowed colours $r,s$ for the terminals of $\Gamma^{*}$, there exist colours $r',s'$}\\
\mbox{such that $r$ is adjacent to $r'$ but not $s'$, and $s$ is adjacent to $s'$ but not $r'$.} 
\end{array}
\end{equation*}

For the claw, the set of allowed colours for the terminals of $\Gamma^*$ in Lemma~\ref{lem:PI} was $\{1,2\}$ and hence the colours in $\{1,2\}$ establish the condition (note the self-loops on colours 1,2).

For the net, the set of allowed colours for the terminals of $\Gamma^*$ in Lemma~\ref{lem:PI} was $\{1,2\}$ and hence the colours in $\{5,6\}$ establish condition \eqref{eq:condH}.

For $S_3$, the set of allowed colours for the terminals of $\Gamma^*$ in Lemma~\ref{lem:PI} was $\{1,2\}$ and hence the colours in $\{4,6\}$ establish condition \eqref{eq:condH}.

For a cycle of length $q\geq 4$, the set of allowed colours for the terminals of $\Gamma^*$ in Lemma~\ref{lem:BP} was $\{1,2\}$ and hence the colours in $\{q,3\}$ establish  condition \eqref{eq:condH}.
\end{proof}

\section{\#BIS-hardness}
We will use the following lemma.
\begin{lemma}\label{lem:BIShardWidom}
For all $\Delta\geq 6$, $\listXcold{P_3^*}{\Delta}$ is $\nBIS$-hard.
\end{lemma}
\begin{proof}
Let $\Delta\geq 6$. We will reduce $\Xcold{P_4}{\Delta}$ to $\listXcold{P_3^*}{\Delta}$. Consider an instance $G$ of $\Xcold{P_4}{\Delta}$. We may assume that $G$ is a bipartite graph with bipartition $(V_1,V_2)$. 

Recall $P_3^*$ from Figure~\ref{fig:2-wrench} and the colours $1,2,3$ there. For each $v\in V_1$, let $S_v=\{1,2\}$ and for each $v\in V_2$, let $S_v=\{2,3\}$. Consider the instance $(G,\bfS)$ of $\listXcold{P_3^*}{\Delta}$. It is immediate to see that list $P_3^*$-colourings of $(G,\bfS)$ are in one-to-one correspondence with $P_4$-colourings of $G$. This completes the reduction, so the result follows from Lemma~\ref{lem:zero}.
\end{proof}

 We now deal with the $\nBIS$-equivalent cases in Theorems~\ref{thm:mainunbounded} and~\ref{thm:main}.

\begin{lemma}\label{lem:BIShard}
\LetH
If $H$ is   not a reflexive complete graph or an irreflexive complete 
bipartite graph then, for all $\Delta\geq 6$,  $\listHcold$ is $\nBIS$-hard. Hence, $\listHcol$ is $\nBIS$-hard.
\end{lemma}

\begin{proof}
The case when $H$ is neither reflexive nor irreflexive is covered 
by Lemma~\ref{lem:(ir)reflexive} (since $\nBIS$ is AP-reducible to $\nSAT$).

So assume next that~$H$ is reflexive 
but not complete. 
We will show that $H$ contains an induced $P_3^*$, which suffices 
by Lemma~\ref{lem:BIShardWidom}.  To find the induced $P_3^*$,
among the non-adjacent pairs of 
vertices $i,j\in V(H)$, choose a pair that minimises the graph distance
between $i$ and~$j$.  Minimality easily implies that the graph distance
between $i$ and~$j$ is in fact~2.  Let $k\in V(H)$ be a vertex that is adjacent
to both $i$ and~$j$.  Then the vertices $\{i,k,j\}$ induce a $P_3^*$. 

Finally, assume that $H$ is irreflexive but that it is not a complete bipartite graph.
\begin{description}
\item[{\bf Case 1}]
If $H$ is bipartite then, among the pairs of 
non-adjacent vertices $i,j\in V(H)$ 
on opposite sides of the bipartition, choose the
pair that minimises graph distance between $i$ and~$j$.  
Minimality easily implies that the graph distance
between $i$ and $j$ is in fact~3.  Thus, $H$ contains an induced
$P_4$, the path of length~$3$. For $\Delta\geq 6$, the problem $\Xcold{P_4}{\Delta}$ is equivalent 
to $\nBIS$ by Lemma~\ref{lem:zero}.  
\item[{\bf Case 2}] If $H$ is not bipartite, it is not a bipartite permutation graph either (trivially). By Lemma~\ref{lem:BPbounded}, for $\Delta\geq 3$, $\listXcold{H}{\Delta}$ is  $\nSAT$-equivalent.
Once again, the result follows since $\nBIS$, like every other approximate counting problem in $\numP$, 
is AP-reducible to $\nSAT$.
\end{description}
\end{proof}

\begin{remark}

Galanis, Goldberg and Jerrum~\cite{NoLife} prove a  result
that is related to, but incomparable to, Lemma~\ref{lem:BIShard}.
Namely, in the absence of lists, and in the absence of degree bounds,
they show that  if $H$  is  not a reflexive complete graph or an irreflexive complete 
bipartite graph then  $\Hcol$  is $\nBIS$-hard. In the unbounded-degree
case, this result is stronger than Lemma~\ref{lem:BIShard} (because no lists are involved). However, 
it does not apply in the bounded-degree case. Also, the proof is quite long 
and technical.
\end{remark}

\section{\#BIS-easiness}\label{sec:BISeasy}

\begin{lemma}\label{lem:BISeasy}
\LetH
If $H$ is an irreflexive bipartite permutation graph or 
a  reflexive proper interval graph, then $\listHcol$ is $\nBIS$-easy.
\end{lemma}

\begin{proof}
The reduction is done in a more general weighted setting by Chen, Dyer, 
Goldberg, Jerrum, Lu, McQuillan and Richerby~\cite{ApproxCSP}:  see the 
proofs of Lemmas 45 and~46 of that article.  However, in the current context,
we can simplify the reduction significantly (eliminating the need for multimorphisms
and other concepts from universal algebra), and we can also extract a slightly 
stronger statement, which will be presented in Corollaries~\ref{cor:RHPi1} and~\ref{cor:RHPi2}.
The target problem for our reduction is $\pnsat$ (see Definition~\ref{def:pnsat}), 
which is $\nBIS$-equivalent by Lemma~\ref{lem:zero}.

First, suppose that $H$ is a connected irreflexive bipartite permutation graph 
whose  
biadjacency matrix~$B$ has $q_1$ rows and $q_2$ columns and
is in staircase form.
Let $A$ be the (permuted adjacency) matrix
$\left( \begin{smallmatrix} B & 0 \\  0 & B^T \end{smallmatrix}\right)$, which is formally
defined as follows.
$$A_{i,j} = \begin{cases} B_{i,j}, & \mbox{if $1\leq i \leq q_1$, $1\leq j \leq q_2$}\\
B_{j-q_2,i-q_1}, & \mbox{if $q_1+1\leq i \leq q_1+q_2$, $q_2+1\leq j \leq q_2+q_1$}\\
0, & \mbox{otherwise.}
\end{cases}$$
Let $q=q_1+q_2$. For each $i\in\{1,\ldots,q\}$, let 
$\alpha_i = \min\{j : A_{i,j} = 1\}$ and let
$\beta_i = \max\{j : A_{i,j}=1\}$.
Since $B$ is in staircase form, so is~$A$, so the sequences $(\alpha_i)$ and $(\beta_i)$ are non-decreasing.
 Let $r_1,\ldots,r_q$ be the colours associated with the rows of~$A$ 
and $c_1,\ldots,c_q$ be the colours associated with the columns of~$A$, in order.
Note that $\{r_1,\ldots,r_q\}$ and $\{c_1,\ldots,c_q\}$ are different permutations of the vertices of~$H$.
 
Suppose that $(G,\mathbf{S})$ is an instance of $\listHcol$. 
Assume without loss of generality that $G$ is bipartite. Otherwise, it has no 
$H$-colourings.
Let $V_1(G)\cup V_2(G)$ be the bipartition of~$V(G)$.
We will construct an instance $\Psi$ of $\pnsat$ such that the number
of satisfying assignments to~$\Psi$ is equal to the number of  
list
$H$-colourings of~$(G,\mathbf{S})$.

The variable set of~$\Psi$ is  
$\bfx=\{x_i^u:u\in V(G)\text{ and }0\leq i\leq q\}$.  For each vertex $u\in V(G)$
we introduce the clauses $(x_0^u)$ and $(\neg x_q^u)$.
Also, for each  $j\in \{1,\ldots,q\}$ we introduce the clause
 $\IMP(x_j^u,x_{j-1}^u)$. 
 Denote by $\Psi_V(\bfx)$ the formula obtained by taking the conjunction 
of all  these clauses.

We will interpret the assignment to the variables in $\bfx$ as an assignment $\sigma$ of colours to
the vertices of~$G$ according to the following rule.
If $u\in V_1(G)$ then
$x_i^u=1$ if and only if $\sigma(u)  = r_j$ for some $j>i$.
If $u\in V_2(G)$ then
$x_i^u=1$ if and only if $\sigma(u) = c_j$ for some $j>i$.
Note that there is a one-to-one correspondence between assignments to $\bfx$ that satisfy the clauses
in $\Psi_V(\bfx)$ and assignments $\sigma$ of colours to the vertices of~$G$.

We now introduce further clauses to enforce the constraint on
colours received by adjacent vertices.  
For each edge $\{u,v\}\in E(G)$ with $u\in V_1(G)$ and $v\in V_2(G)$,
and for each
 $i\in \{1,\ldots,q\}$, we add the clauses
$
\IMP(x_{i-1}^u,x_{\alpha_i-1}^v)$ and
$\IMP(x_{\beta_i}^v,x_i^u)$. 
Denote by $\Psi_E(\bfx)$ the formula obtained by taking the conjunction of all of these clauses.
 
We next argue that there is
a bijection between  $H$-colourings of~$G$ 
and satisfying assignments to~$\Psi_V(\bfx) \wedge \Psi_E(\bfx)$.  

In one direction, suppose $\sigma$ is an  $H$-colouring of~$G$.    
We wish to show that all clauses in $\Psi_E(\bfx)$ are satisfied. 
Consider an edge 
$\{u,v\}\in E(G)$ with $u\in V_1(G)$ and $v\in V_2(G)$.
\begin{itemize}
\item Consider the corresponding clause 
$\IMP(x_{i-1}^u,x_{\alpha_i-1}^v)$.
The clause is satisfied unless $x_{i-1}^u=1$, so suppose $x_{i-1}^u=1$.
Then by the interpretation of assignments, $\sigma(u)=r_j$ for some $j\geq i$.
Since $\sigma$ is an $H$-colouring, this implies that
$\sigma(v)=c_k$ for some $k\geq \alpha_i$.
But by the interpretation of assignments, this means that $x_{\alpha_i-1}^v=1$, so the clause is satisfied.
\item Consider the other corresponding clause 
$\IMP(x_{\beta_i}^v,x_i^u)$. 
Suppose that $x_{\beta_i}^v=1$ (otherwise the clause is satisfied).
Then by the interpretation of assignments, $\sigma(v)=c_k$ for some $k> \beta_i$.
Since $\sigma$ is an $H$-colouring, this implies that
$\sigma(u) = r_j$ for some $j>i$, which implies by the interpretation of assignments
that $x_i^u=1$ so the clause is satisfied.
\end{itemize}

In the other direction, suppose $\Psi_V(\bfx) \wedge \Psi_E(\bfx)$ is satisfied.
Consider an edge 
$\{u,v\}\in E(G)$ with $u\in V_1(G)$ and $v\in V_2(G)$
and suppose that $\sigma(u)=r_i$.
\begin{itemize}
\item In the corresponding assignment $x_{i-1}^u=1$ so 
by the clause $\IMP(x_{i-1}^u,x_{\alpha_i-1}^v)$
we have $x_{\alpha_i-1}^v=1$ so $\sigma(v) = c_k$ for some $k\geq \alpha_i$.
\item In the corresponding assignment $x_{i}^u=0$ so
by the clause $\IMP(x_{\beta_i}^v,x_i^u)$, $x_{\beta_i}^v=0$, so $\sigma(v) = c_k$ for
some $k\leq \beta_i$.
\end{itemize}
We conclude that the colours $\sigma(u)$ and $\sigma(v)$ are adjacent in~$H$.
This holds for every edge, 
so $\sigma$ is an $H$-colouring of~$G$.

Finally, we add clauses to deal with lists. 
 A colour assignment $\sigma(u)=r_i $ with $u\in V_1(G)$
   is uniquely 
characterised by  
$x_{i-1}^u=1$ and $x_i^u=0$.
So we can eliminate the
possibility of $\sigma(u)=r_i$ by introducing the clause $\IMP(x_{i-1}^u,x_i^u)$.
A similar clause will forbid  a vertex $v\in V_2(G)$ to receive colour~$c_j$.
Let $\Psi_L(\bfx)$ be the conjunction of all such clauses, arising from the lists in $\mathbf{S}$.   
Let $\Psi(\bfx)= \Psi_V(\bfx) \wedge \Psi_E(\bfx) \wedge \Psi_L(\bfx)$.

Then the  list $H$-colourings of $(G,\bfS)$ are 
in bijection with the satisfying assignments to $\Psi(\bfx)$.
This concludes the case where $H$ is an irreflexive bipartite permutation graph.

The situation where $H$ is a reflexive proper interval graph is exactly the same
except that we can just take the adjacency matrix $A$ of~$H$ to be in staircase form,
so $\{r_1,\ldots,r_q\}$ is the same permutation as $\{c_1,\ldots,c_q\}$. We do not require $G$ to
be bipartite so the interpretation of the assignment of the variables in $\bfx$ as an assignment $\sigma$
of colours to the vertices of~$G$ is the same for all vertices in~$G$.
\end{proof}

As remarked in  Section~\ref{sec:intro}, we can slightly strengthen the statement of Lemma~\ref{lem:BISeasy}.
For this, we will need the definition of the complexity class $\RHPi$,
from \cite[Section 5]{DGGJ}, which builds
on the logical framework of Saluja et al~\cite{Saluja}.  
A \emph{vocabulary} is a finite set of relation symbols. 
These  are used to define
an instance of a $\numP$ problem.
In the case of $\nBIS$,  the relevant vocabulary consists of a binary relation~$\sim$ 
(representing adjacency in the input graph~$G$) and a unary relation~$L$
(representing one side of the bipartition --- for concreteness, the ``left'' side $V_1(G)$).
  
A problem in $\numP$ can be represented by a first-order sentence using the relevant vocabulary.
The sentence may use variables to represent elements of the universe, it may use 
relations in the vocabulary, and it may also use new some relation symbols 
$X_0, \ldots, X_q$ (for some~$q$).
(The fact that the sentence is ``first-order'' just means that it may quantify over variables, but
not over relations.)

The input to the $\numP$ problem is a structure consisting of
a universe together with interpretations of all relations in the vocabulary.
The problem is to count the number of
ways that the interpretations can be extended to interpretations of the
new relation symbols $X_0,\ldots,X_q$ and of any free variables in the sentence
to obtain a model of the sentence. 
Note that the input is an arbitrary structure  over the vocabulary but, for example, only
some structures over the vocabulary~$\{\sim, L\}$ correspond to  undirected bipartite
graphs, so in the logical framework we define $\nBIS$ as follows.
\begin{description}
\item[Name] $\nBIS$.
\item[Instance]  A structure consisting of a finite universe~$V$ and interpretations 
of the binary relation~$\sim$ and the unary relation~$L$.
\item[Output]  If the structure corresponds to an undirected bipartite graph~$G$
where $L$ represents one side of the bipartition then the output
is the number of independent sets of~$G$. Otherwise,  it is~$0$.
\end{description}
\begin{remark} Since it is easy to check in polynomial time whether a structure over $\{\sim, L\}$ does
represent an undirected bipartite graph with bipartition given by~$L$,
this version of $\nBIS$ is AP-interreducible with the usual version.
 \end{remark}

In the class $\RHPi$, syntactic restrictions are placed on the first-order sentence
representing the problem. It  
consists of universal quantification (over variables) followed by an unquantified CNF formula
which satisfies the special property that each clause has at most one occurrence of
an unnegated relation symbol from $X_0,\ldots,X_q$ and at most one occurrence of
a negated relation symbol from $X_0,\ldots,X_q$.

It is known   that $\nBIS$, and many other problems that are $\nBIS$-equivalent 
are contained in $\RHPi$, and in fact that they are complete for $\RHPi$ with respect to AP-reductions.
To illustrate the definitions, it is perhaps useful to give a  simple example illustrating the fact that
$\nBIS$, as defined above, is in $\RHPi$. A similar example appears in~\cite{DGGJ}.
A suitable sentence using the vocabulary $\{\sim, L\}$ is
\begin{align*}
\forall u,v \quad &
\left(L(u)\wedge u\sim v \wedge X_0(u)\implies X_0(v)\right)
\wedge \\
& \left(u \sim v \implies v \sim u \right)
\wedge
\left(\neg(u \sim v) \vee L(u) \vee L(v) \right)
\wedge
\left(\neg(u \sim v) \vee \neg L(u) \vee \neg L(v) \right). \end{align*}
Note that  the clauses 
$$\left(u \sim v \implies v \sim u \right)
\wedge
\left(\neg(u \sim v) \vee L(u) \vee L(v) \right)
\wedge
\left(\neg(u \sim v) \vee \neg L(u) \vee \neg L(v) \right)
 $$
ensure that the sentence has no models unless the input corresponds to an undirected
bipartite graph~$G$ where the relation~$L$ corresponds to one side of the bipartition.
When  this is the case,
each interpretation of the new unary relation symbol~$X_0$  corresponds to an independent set of~$G$
in the sense that 
$\{ u \in L \cap X_0 \} \cup \{u \in \overline{L} \cap \overline{X_0}\}$ is an independent set,
so independent sets are in one-to-one correspondence with interpretations of~$X_0$.
 
Having defined $\RHPi$, we are now ready to strengthen Lemma~\ref{lem:BISeasy}
to show that, if $H$ is an irreflexive bipartite permutation graph or
a reflexive proper interval graph, then $\listHcol\in \RHPi$.
In order to proceed, we must define versions of these  in the logical framework, as we did for~$\nBIS$.
First, suppose that $H$ is a connected irreflexive bipartite permutation graph with $q$ vertices.
Our vocabulary will consist of the vocabulary $\{\sim, L\}$,  together with 
$q$ unary relations $U_1, \ldots, U_{q}$.
The intention is that $U_i$ will represent the vertices of~$G$
that are allowed colour~$i$.

\begin{description}
\item[Name] $\IBPlistHcol$.
\item[Instance]  A structure consisting of a finite universe~$V$ 
and interpretations of the relations in the vocabulary $\{\sim, L,U_1, \ldots, U_{q}  \}$. 
\item[Output] If the structure consisting of $V$, $\sim$ and $L$
corresponds to an undirected bipartite graph $G$ where
$L$ represents one side of the bipartition, then  the output is
the number of  list $H$-colourings of $(G,\bfS)$ where, for each $v\in V$,
$S_v = \{ i : v\in U_i\}$ and  
$\bfS=\{S_v\}$.   Otherwise,  the output is~$0$. \end{description} 
Since it is easy to check in polynomial time whether a structure over
$\{\sim,L\}$ does represent an undirected bipartite graph, with bipartition given by~$L$,
and since $H$ is irreflexive and bipartite (so $G$ has no $H$-colourings unless it is also bipartite),
the problem $\IBPlistHcol$ is AP-interreducible with  the usual version.

\begin{corollary}\label{cor:RHPi1}
\LetH
If $H$ is an irreflexive bipartite permutation graph  then the problem $\IBPlistHcol$
is in $\RHPi$.
\end{corollary}

\begin{proof}
The proof  is essentially a translation 
of the reduction from Lemma~\ref{lem:BISeasy}
into the logical setting of $\RHPi$. 
The sentence representing the   problem $\IBPlistHcol$ is of of the form 
$\forall u,v \, \Phi(u,v)$ where $\Phi(u,v)$ is a conjunction of clauses.
As in the $\nBIS$ example, we include the clauses
$$\left(u \sim v \implies v \sim u \right)
\wedge
\left(\neg(u \sim v) \vee L(u) \vee L(v) \right)
\wedge
\left(\neg(u \sim v) \vee \neg L(u) \vee \neg L(v) \right)
 $$
  to ensure that the output
is~$0$ unless the structure consisting of $V$, $\sim$ and $L$ corresponds to
an undirected bipartite graph~$G$ with bipartition given by~$L$. Suppose that this is so.

For each $0\leq i\leq q$, we 
introduce a new
unary relation $X_i$ corresponding (collectively) to the variables  with subscript~$i$ in 
the variable set $\bfx$ from Lemma~\ref{lem:BISeasy}.
Then the remaining clauses of $\Phi(u,v)$ 
come directly from the formula~$\Psi(\bfx)$.
So  from $\Psi_V$ 
we have the following clauses (which are inside the universal quantification over~$u$):
$(X_0(u))$, $\neg(X_q(u))$ and, for each $j\in\{1,\ldots,q\}$,
 $\IMP(X_j(u),X_{j-1}(u))$.

Recall that for each edge $\{u,v\}\in E(G)$ with $u\in V_1(G)$ and $v\in V_2(G)$,
$\Psi_E(\bfx)$ contains the clauses
$\IMP(x_{i-1}^u,x_{\alpha_i-1}^v)$
and $\IMP(x_{\beta_i}^v,x_i^u)$. 
In $\Phi$ these becomes the clauses 
$$ L(v)  \vee \neg (u\sim v) \vee \IMP(X_{i-1}(u), X_{\alpha_i-1}(v))$$
and
$$ L(v)  \vee \neg (u\sim v) \vee \IMP(X_{\beta_i}(v),X_i(u)).$$

  Finally,   we eliminated the possibility of 
  $\sigma(u)=i$ by adding to $\Psi_L(\bfx)$ the clause 
 $\IMP(x_{i-1}^u,x_i^u)$.
 In $\Phi$ this becomes
 $U_i(u) \vee \IMP(X_{i-1}(u),X_i(u))$. 
  
 The proof of Lemma~\ref{lem:BISeasy} guarantees that
  the models of~$\Phi$ correspond to  the list $H$-colourings of 
$(G,\mathbf{S})$, as desired. Also, each clause of~$\Phi$ uses at most one negated
relation   and at most one unnegated relation from 
the set $\{X_0,\ldots,X_q\}$.
Thus, 
 $\IBPlistHcol$ is
  in $\RHPi$.
\end{proof}

The case when $H$ is a reflexive proper interval graph is similar, but easier.
Let $q$ be the number of vertices of~$H$.
We define the problem in the logical framework using the binary relation~$\sim$ and the unary
relations~$U_1,\ldots,U_q$ as follows.
\begin{description}
\item[Name] $\RPIlistHcol$.
\item[Instance]  A structure consiting of a finite universe~$V$ 
and interpretations of the relations in the vocabulary $\{\sim, U_1, \ldots, U_{q}  \}$. 
\item[Output] If the structure consisting of $V$ and $\sim$  
corresponds to an undirected   graph $G$, then  the output is
the number of  list $H$-colourings of $(G,\bfS)$ where, for each $v\in V$,
$S_v = \{ i : v\in U_i\}$ and  
$\bfS=\{S_v\}$.   Otherwise,  the output is~$0$. \end{description} 
Then following corollary follows directly from the proof of Lemma~\ref{lem:BISeasy}, in the same was as
Corollary~\ref{cor:RHPi1}.

\begin{corollary}\label{cor:RHPi2}
\LetH
If $H$ is a reflexive proper interval graph  then the problem $\RPIlistHcol$
is in $\RHPi$.
\end{corollary}

\section{A counterexample}\label{sec:counterexample} 
The situation that we have studied in this paper is characterised by having
hard interactions between pairs of adjacent spins (a pair is either allowed
or it is disallowed) and hard constraints on individual
spins (again, a spin is either allowed at a particular vertex or it is disallowed).  
Our results apply both in the degree-bounded case and in the unbounded-degree case. 
In the latter case,
earlier work treated the situation with weighted interactions and weighted spins.  
The characterisations derived in these weighted scenarios (see, e.g. \cite[Thm~1]{PNAS})
have a similar feel to the trichotomy that we have presented in Theorem~\ref{thm:mainunbounded}.  
We may wonder whether, in the unbounded case, at least,  there is a common generalisation.  
That is, in the unbounded case, does the trichotomy of \cite{PNAS} survive if weights on spins are replaced by lists? 
The answer is no.  
There are examples of weighted spin systems with just $q=2$ spins whose partition
function is $\nSAT$-hard to approximate with vertex weights but efficiently 
approximable (in the sense that there is an FPRAS) with lists instead of weights.

Here is one such example.
Following Li, Lu and Yin~\cite{LiLuYinEarlier}, define the 
interaction matrix $A=(a_{ij}:0\leq i,j\leq1)$ 
by $A=(\begin{smallmatrix}0&1\\1&2\end{smallmatrix})$,
and the partition function associated with an instance~$G$ by 
\begin{equation}\label{eq:pf}
Z_A(G)=\sum_{\sigma:V(G)\to\{0,1\}}\,\,\prod_{\{u,v\}\in E(G)}a_{\sigma(u),\sigma(v)}.
\end{equation}
This is the partition function of a variant of the independent set model,
which  instead of defining 
the interaction between spin~$1$ and spin~$1$  
(two vertices that are out of the independent set) to be~$1$, defines this interaction 
weight to be~$2$.

Li, Lu and Yin \cite[Theorem 21]{LiLuYinEarlier} show that
Weitz's self-avoiding walk algorithm \cite{Weitz} gives an FPTAS for $Z_A(G)$.
Also,
Weitz's correlation decay algorithm \cite{Weitz}
can accommodate 
lists.  Indeed, the 
construction of the self-avoiding walk tree
relies on being able to ``pin'' colours at individual vertices.  So the partition
function ~(\ref{eq:pf}) remains easy to approximate (in the sense
that there is an  FPTAS) even in the presence of lists.  In contrast, the approximation problem becomes  $\nSAT$-hard if arbitrary 
weights are allowed.  
Indeed, by weighting 
spin~$0$ 
at each vertex $u\in V(G)$ 
by $2^{d(u)}$, where $d(u)$ is the degree of~$u$,
we recover the usual independent set partition function, which 
is  $\nSAT$-equivalent (Lemma~\ref{lem:zero}).  (The same fact can be read off from general results 
in many papers, including \cite[Thm~1]{PNAS}.)
Thus, even in the unbounded case,  the dichotomies presented in 
\cite[Thm 1]{PNAS} and \cite[Thm~6]{ApproxCSP}
do not hold with lists in place of weights.
So even in the unbounded-degree case, it was necessary to explicitly analyse  list homomorphisms  
in order to derive
precise characterisations  quantifying the problem of approximately counting these.

\section{Appendix:  characterisations of hereditary graph classes}\label{sec:app}

In this appendix we justify the graph-theoretic claims made in Section~\ref{sec:prelim},
by reference to the literature.

\begin{itemize}
\item {\it Excluded subgraphs for bipartite permutation graphs.}
This characterisation is due to K\"ohler \cite{Koehler}, building on work of Gallai~\cite{Gallai}.
A graph is said to be AT-free if it does not contain 
a certain kind of induced substructure known as an asteroidal triple.
The result we need follows from \cite[Lemma 1.46]{Koehler} together 
with the fact that a bipartite graph is AT-free if and only if it is a permutation graph.
See the discussion following \cite[Lemma 1.46]{Koehler} for a statement of the latter fact, 
together with a proof of the ``only if'' direction.
The ``if'' direction is not explicitly proved in \cite{Koehler} but can be 
deduced by comparing the excluded subgraphs in \cite[Fig.~1.6]{Koehler} 
and \cite[Fig.~1.16]{Koehler}, and noting that all the excluded subgraphs 
for bipartite AT-free graphs are also excluded subgraphs for bipartite permutation graphs.  

\item {\it Excluded subgraphs for proper interval graphs.}
This characterisation is due to Wegner~\cite{Wegner} and Roberts~\cite{Roberts}.
It follows from combining \cite[Thm 3(a,b)]{RobertsIndifference}  
with \cite[Thm 6(a,c)]{RobertsIndifference}.  

\item {\it Matrix characterisation for bipartite permutation graphs.}
This characterisation follows indirectly from those given by Spinrad et al.~\cite{SpinradEtAl}.
The equivalence is proved as Lemma 7 of \cite{DJM}.

\item {\it Matrix characterisation for proper interval graphs.}
Mertzios~\cite{Mertzios} shows that an $n$-vertex graph~$G$ is a proper interval graph 
if and only if it is the intersection graph of a set of $n$~intervals of a restricted form,
called the ``Stair Normal Interval Representation'' (SNIR) of~$G$.  
These intervals have distinct left endpoints in the set $\{0,1,\ldots,n-1\}$;
moreover, when the intervals are placed in order of increasing left endpoint,
the right endpoints are in monotone non-decreasing order. 
By~\cite[Lemma 3]{Mertzios}, any proper interval graph can be transformed to SNIR form,
after which the adjacency matrix (with ones on the diagonal) is in staircase form.  Although it is not mentioned explicitly
in~\cite{Mertzios}, it is clear that an adjacency matrix (with ones on the diagonal) in staircase form corresponds to an SNIR form by considering the intervals induced by the non-zero elements in the rows of the upper triangular portion of the matrix, and hence to a proper interval graph.
\end{itemize}

\bibliographystyle{plain}
\bibliography{\jobname}

\end{document}